\documentclass[a4paper,english,12pt,reqno,oneside]{amsart}

\usepackage{amsmath, amsthm, amssymb}
\usepackage{a4, babel, graphicx, listings, fancyhdr, verbatim}
\usepackage{mathrsfs} 
\usepackage{anysize} 
\usepackage{natbib}
\usepackage[format = hang]{caption}                              
\usepackage{verbatim}

\usepackage{a4}
\usepackage[latin1]{inputenc}
\usepackage{amsmath, amsthm}
\usepackage{graphicx}
\usepackage{amssymb}
\usepackage{verbatim}
\usepackage{listings}
\usepackage{alltt}
\usepackage{babel}
\usepackage{mathrsfs}
\usepackage{ae}
\usepackage{natbib}

\newtheoremstyle{theorem}
{10pt} 
{10pt} 
{\sl} 
{\parindent} 
{\bf} 
{. } 
{ } 
{} 

\theoremstyle{theorem}

\newtheorem{proposition}{Proposition}

\newtheorem{remark}{Remark}

\def\beq{\begin{eqnarray}}
\def\eeq{\end{eqnarray}}

\def\beqn{\begin{eqnarray*}}  
\def\eeqn{\end{eqnarray*}}

\def\E{{\rm E}}
\def\Var{{\rm Var}}
\def\cov{{\rm cov}}
\def\dd{{\rm d}}
\def\N{{\rm N}}
\def\fic{{\rm FIC}}
\def\bsq{{\rm bsq}}

\def\Pr{P}
\def\pr{{\rm pr}}

\def\risk{{\rm risk}}
\def\afic{{\rm AFIC}}

\def\quadandquad{\quad {\rm and} \quad}
\def\arr{\rightarrow}
\def\hatt{\widehat}
\def\hati{\widehat}

\def\tildi{\widetilde}

\def\eps{\varepsilon}
\def\half{\hbox{$1\over2$}}

\def\rootn{\sqrt{n}}

\def\midd{\,|\,}
\def\tr{{\rm t}}

\def\argmin{{\rm argmin}}
\def\argmax{{\rm argmax}}

\def\true{{\rm true}}

\def\Tr{{\rm Tr}}

\def\mse{{\rm mse}}

\def\corr{{\rm corr}}

\usepackage[a4paper,textheight=24.0truecm,textwidth=16.0truecm,top=2.7truecm]{geometry}
\numberwithin{equation}{section} 
\numberwithin{figure}{section}
\numberwithin{table}{section}
\usepackage{hyperref}

\newcommand{\R} {\mathbb{R}}
\newcommand{\pa}{{\rm pm}}
\newcommand{\np}{{\rm np}}
\title[FIC for stationary time series]{Parametric or nonparametric: 
   the FIC approach for stationary time series}
\date{December 2015}

\begin{document}


\maketitle

\centerline{\bf Gudmund Hermansen, Nils Lid Hjort, Martin Jullum}  

\medskip 
\centerline{\bf Department of Mathematics, University of Oslo} 

\medskip 
\centerline{\bf December 2015} 

\begin{abstract}
We seek to narrow the gap between parametric and 
nonparametric modelling of stationary time series 
processes. The approach is inspired by recent 
advances in focused inference and model 
selection techniques. 
The paper generalises and extends recent work 
by developing a new version of the focused information 
criterion (FIC), directly comparing the performance of
parametric time series models with a nonparametric alternative.
For a pre-specified focused parameter,
for which scrutiny is considered valuable,
this is achieved by comparing the mean squared error of
the model-based estimators of this quantity. In particular, this yields
FIC formulae for covariances or correlations 
at specified lags, for the probability of reaching
a threshold, etc. Suitable weighted average versions,
the AFIC, also lead to model selection strategies for
finding the best model for the purpose of estimating
e.g.~a sequence of correlations.
\end{abstract}



\section{Introduction and summary}
\label{section:intro}

The focused information criterion (FIC) was introduced 
in \cite{ClaeskensHjort03} and is 
based on estimating and comparing the accuracy 
of model-based estimators
for a chosen focus parameter. This focus, say $\mu$, 
ought to have a clear statistical interpretation across candidate models. 
For a given candidate model, $\mu$ is traditionally 
expressed as a function of this model's parameters.
In general, the focus parameter can be any sufficiently
smooth and regular function of the underlying model 
parameters, or more generally its spectral distribution. This includes quantiles, regression coefficients, 
a specified lagged correlation, but also various types 
of predictions and data dependent functions, to
name some; see \cite{HermansenHjort15c} for a more complete list
and discussion of valid focus parameters for time series models.

Suppose there are candidate models $M_1,\ldots,M_k$,
leading to focus parameter estimates 
$\hatt\mu_1,\ldots,\hatt\mu_k$, respectively. The 
underlying idea leading to the FIC is to estimate 
the mean squared error (mse) of $\hatt\mu_j$
for each candidate model and then select the model that achieves
the smallest value. The mse in question is
$$\mse_j=\E\,(\hatt\mu_j - \mu_{\true})^2
   = \textrm{bias}(\hatt\mu_j)^2 +  \Var\,\hatt\mu_j, $$
comprising the variance and the squared bias
in relation to the true parameter value $\mu_\true$.
Thus the FIC consists of finding ways of assessing,
approximating and then estimating the $\mse_j$ for
each candidate model. The winning model is the
one with smallest $\hatt \mse_j$. How this may be done depends
on both the candidate models and the focus parameter,
as well as on other characteristics of the underlying situation.
The FIC apparatus hence leads to different types of 
formulae in different setups;
see \citet[Ch.~5 \& 6]{ClaeskensHjort08} 
for a fuller discussion and illustrations of
such criteria for selection among parametric models.


Most FIC constructions have been derived by relying on a
suitably defined local misspecification framework, see again 
\citet[Ch.~5 \& 6]{ClaeskensHjort08}. 
In such a framework the true model is assumed to gradually shrink
with the sample size, starting from the biggest `wide' model and hitting 
the simplest `narrow' model in the limit. In addition, and all candidate models need 
to lie between these two model extremes. 
In the various data settings, such frameworks typically result in squared 
biases and variances of the same asymptotic order, motivating certain 
approximation formulae for the $\hatt \mse_j$ in question. 
In \cite{HermansenHjort15c} such a framework is used to derive FIC machinery 
for choosing between parametric time series models within broad 
classes of time series models. 
 See Section \ref{sec:local} for some further remarks.

The aim of the present paper is to derive FIC machinery 
which will justify comparison and selection
among both parametric and nonparametric candidate models.
The derivation will be somewhat different from that of
\citet{ClaeskensHjort03} and \cite{HermansenHjort15c} in that we do not rely on a certain
local misspecification framework. We rather take a more direct approach
following reasoning similar to the development
of \cite{JullumHjort15}, where focused inference and model selection 
among parametric and nonparametric models are developed 
for independent observations.
By including a nonparametric candidate among the parametric 
models, we will in particular be able to detect whether 
our parametric models are off-target. This FIC construction, with a nonparametric alternative, 
therefore has a built-in insurance mechanism against 
poorly specified parametric candidates. When one or more
parametric models are adequate, such are selected 
as they typically have lower variance.

Though our methods will be extended to more general setups later, 
we start our developments with the class of zero-mean stationary 
Gaussian time series processes. Let $\{Y_t\}$ be such a process. Then 
the dependency structure, which in such cases determines 
the entire model, is completely specified by the corresponding 
covariance function $C(k) = \cov(Y_t, Y_{t+k})$, defined for all lags
$k = 0, 1, 2, \ldots$. Here we will, for mathematical convenience, 
work with the frequency representation, where the covariance function 
$C(k)$ can be represented by a unique spectral distribution $G$ such that 
\beq
\label{covariance.function}
	C(k)= \int_{-\pi}^{\pi} e^{ik\omega}\,\dd G(\omega)
    = 2 \int_{0}^{\pi} \cos(k\omega) g(\omega) \,\dd \omega,
\eeq 
provided the corresponding spectral distribution 
$G$ has a continuous and symmetric density $g$.
See among others \cite{Brillinger75}, \cite{Priestley81} or \cite{Dzhaparidze86}
for a general introduction to time series modelling in the frequency domain. 
When necessary, we will write $C_{g}$ to indicate that this is the covariance 
indexed by the spectral density $g$.
Note also that we can obtain the spectral density as the Fourier transform
of the covariance function.



The types of parametric models we will consider are typically the 
classical autoregressive (AR), moving average (MA) and 
the mixture (ARMA), all of which have clear and well defined 
corresponding spectral densities; see e.g.~\cite{Brockwell91}
for an introduction to time series modelling with such models. Note 
that the theory developed here is general, and that there is nothing 
other than convenience that restricts us to these particular classes of parametric
models. For an observed series $y_{1}, \ldots, y_{n}$, 
the raw periodogram
\beq
\label{eq:periodogram}
	I_{n}(\omega) = \frac{1}{2 \pi n } \bigg | \sum_{t = 1}^{n}
	 y_{t} \exp (i \omega t ) \bigg |^{2}, \quad \textrm{ for } -\pi \le \omega < \pi,
\eeq
will be our favourite nonparametric model for the underlying spectral 
density. The main reason for not considering variations of smoothed or 
tapered periodogram estimators is that we are interested in focus parameters 
that involves functions of the integrated spectrum, which essentially 
is a type of smoothing, rendering the pre-smoothing of the raw periodogram  
less critical and often unnecessary.

We will start out considering a class of focus functions 
of the type
\beq
\label{eq:focus}
\mu(G; h_0)=\int_{-\pi}^{\pi} h_0(\omega)\,\dd G(\omega),
\eeq
where $h_0$ is a piecewise continuous and bounded function 
on $[-\pi,\pi]$, with potentially 
a finite number of jump discontinuities. 
This class includes e.g.~the covariance 
function, which is easily seen from (\ref{covariance.function}) above, 
and allows studying specific parts 
of the spectral density by using indicator functions; 
see also \cite{Gray06} for further illustrations
involving quantities of type (\ref{eq:focus}).

\begin{figure}[ht]
\centering
\includegraphics[width=0.90\textwidth]{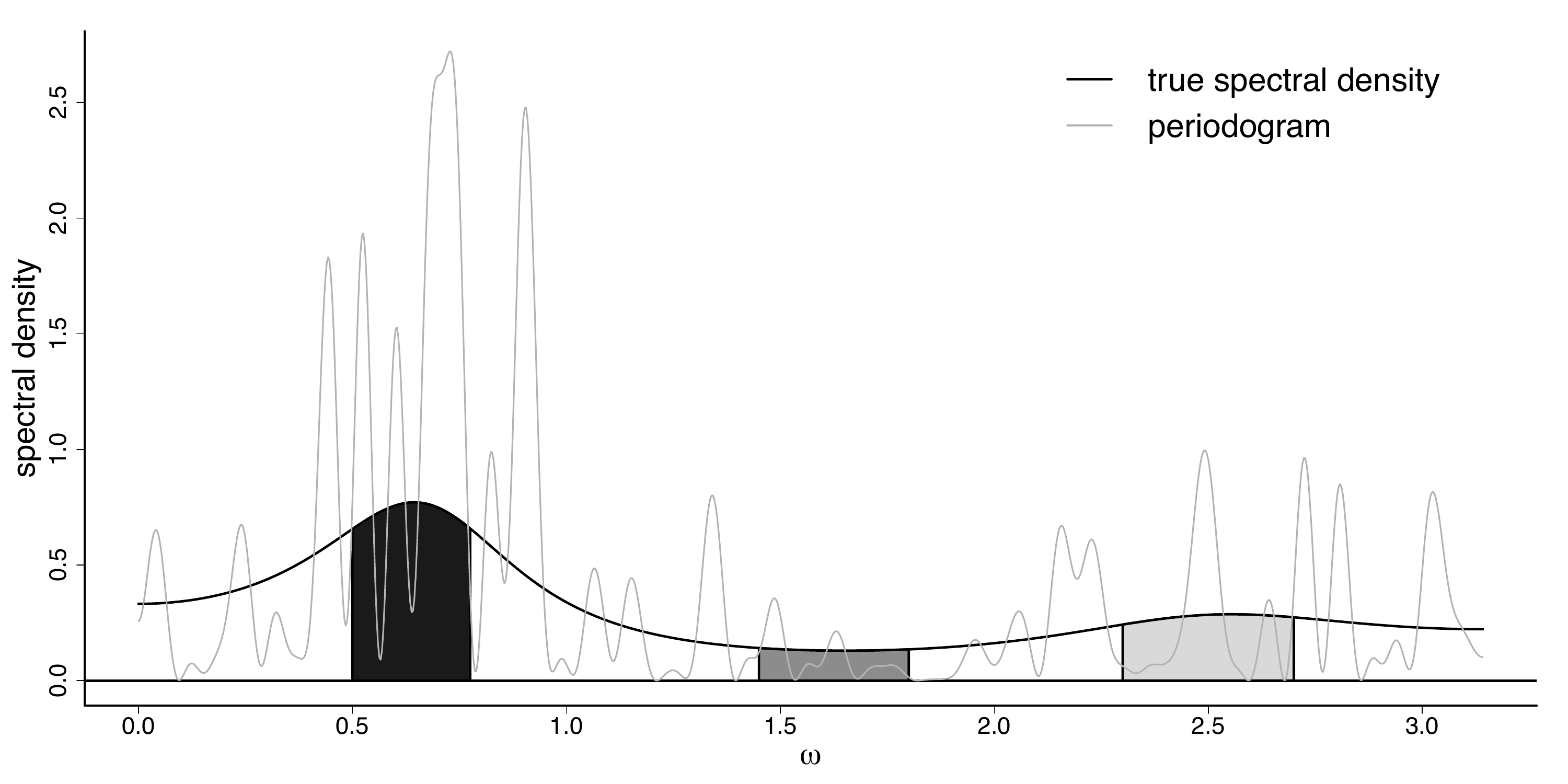}
\caption{
The true spectral density and the raw periodogram from a  simulated 
autoregressive time series of order $4$, with length $n = 100$ and 
 parameters $\rho = (0.2, 0.2, -0.1, -0.2)$ and 
$\sigma = 1.30$.
The shaded regions corresponds to three different 
focus parameters, namely, the integrated spectrum (or total energy) 
over that particular region. 
}
\label{figure:intro:data}
\end{figure}

Finding the best model to estimate the integrated spectrum 
(or total power/energy) over a specific region, may be an 
interesting and important applications in several areas of 
research; like pharmacology, astronomy, oceanography 
and in the interpretation of seismic data. The reason is that 
in all of these situations the observed time series 
is converted into the associated spectra, where the processed 
spectral density and especially the energy over certain regions 
of frequencies, have clear interpretations. For example, in pharmacology 
the spectrum of EEG/ERP signals may be used to quantify 
certain brain functions, indicating e.g.~the effect of a potential drug. 
In such applications, the different models 
may not always have clear interpretations
as time series, per se. The FIC is nevertheless
able to rank the fitted models in terms of 
estimated precision of estimates, for the
focus parameter in question. 
This general idea and particular usage of the FIC 
is illustrated in Figures \ref{figure:intro:data} and \ref{figure:intro:fic}
using simulated data from an autoregressive model of order 4, for focus parameters
$$
	\mu_{j} = \int_{0}^{\pi} I(a_{j} \le \omega < b_{j}) g(\omega) \, \dd \omega = G(b_j)-G(a_j), 
$$
for $j = 1, 2$ and 3, for the corresponding intervals $(a_j, b_j) \subset [0,\pi)$;
which are marked by the shaded regions in Figure \ref{figure:intro:data}.
The candidate models are the autoregressive models of order 0--4 and a nonparametric 
alternative based on integrating the raw periodogram (\ref{eq:periodogram}). The
AR-model of order 0 corresponds to the independence model. 
Here, the FIC works well: For each focus parameter it prefers models that 
all results in estimates that are reasonably close to the true value; which in terms 
if rmse (and absolute deviation from the truth) is not always the nonparametric or true model of order 4. 
Moreover, this example also illustrates a second and important concept, namely, 
that one and the 
same model is not necessarily best for all focus parameters. Note that the FIC prefers 
an AR(3), AR(4) and AR(1) for the respective regions 1, 2, 3.

\begin{figure}
\centering
\includegraphics[width=0.90\textwidth]{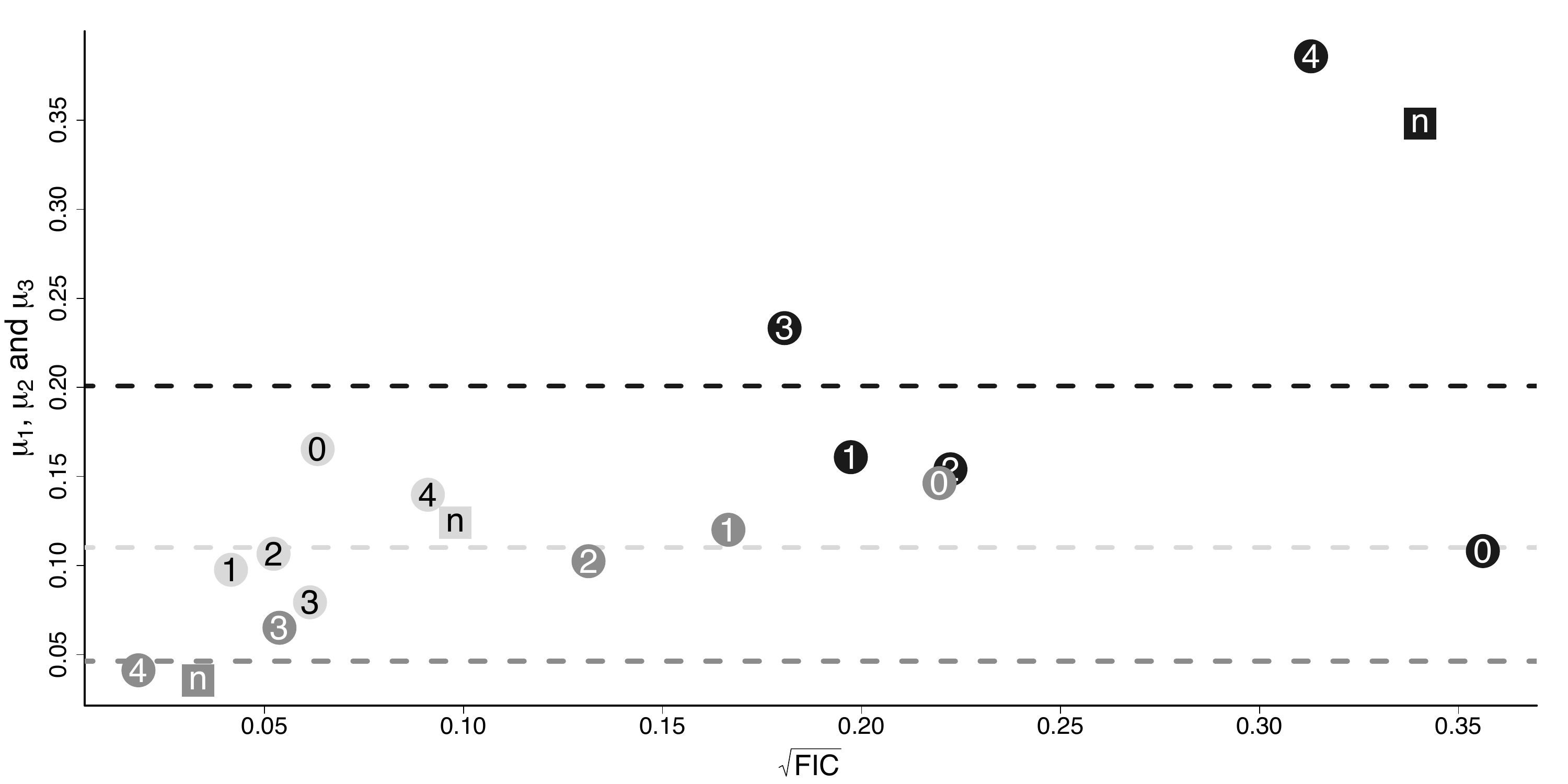}
\caption{
The horizontal lines indicate the true spectral density over the 
three shaded regions (of the same colour) shown in Figure 
\ref{figure:intro:data}; the three focus parameters $\mu_1, \mu_2$ and $\mu_3$.
The corresponding coloured dots 
show the performance, in terms of the root of the FIC score for 
the nonparamteric model based on the periodogram (n) and 
the autoregressive models of order 0--4, where 0 represent 
the model with independent. 
}
\label{figure:intro:fic}
\end{figure}

A class of focus parameters wider than that of (\ref{eq:focus}) takes focus parameters of the form
\beq
\label{eq:focus2}
\begin{aligned}
\mu(G;h,H)&=H(\mu(G;h_1),\ldots,\mu(G;h_k)) \\
   &=H\Bigl(\int_{-\pi}^\pi h_1(\omega)\,\dd G(\omega),\ldots,\int_{-\pi}^\pi h_k(\omega)\,\dd G(\omega)\Bigr), 
\end{aligned}
\eeq 
for a $k$-dimensional vector function $h(\omega)=(h_1(\omega),\ldots,h_k(\omega))^{\tr}$, where each of the $h_j$ is of the above type, 
and $H(x_1,\ldots,x_k)$ a continuously differentiable function of the $x_j=\mu(G;h_j), j=1,\ldots,k$. 
The direct correlations 
$$\corr(Y_t,Y_{t+k})={\cov(Y_t,Y_{t+k})\over \sigma^2}
   ={C(k)\over C(0)}
   ={\int_0^\pi \cos(k\omega)\,\dd G(\omega)
   \over \int_0^\pi\dd G(\omega)}, $$
for example, are of type (\ref{eq:focus2}). Another class 
of estimands captured by (\ref{eq:focus2}) are conditional 
threshold probabilities, say  
$\Pr\{Y_{n+1}\ge y\midd Y_{n}=y_n,\ldots, Y_{n-k}=y_{n-k}\}$, 
as these are functions of the $(k+1)\times(k+1)$ covariance matrix
for $(Y_{n-k},\ldots,Y_n,Y_{n+1})$. 
Later results will allow us to reach FIC formulae 
for this more general class.

In Section \ref{section:thefrequencyapproach} we provide 
a brief overview of some standard results needed 
to obtain good estimates for various mean squared error
quantities. Among other aspects we need properties 
of maximum likelihood- or Whittle approximated estimators outside the model, and some large-sample results regarding 
the periodogram. Then in Section \ref{section:parametricornonparametric}
we motivate and develop such mean squared error estimators,
leading to FIC formulae. 
In Section \ref{section:detrending} we show that under certain conditions, 
a detrended time series may be handled by our FIC scheme as if it was the original time series.
In Section \ref{section:afic} we extend the FIC methodology by
deriving an average weighted focused information criterion
which aims at selecting the best model for estimating a full set of
focus parameters, possibly weighted to reflect their relative importance for the analysis.
In Section \ref{section:performance} we discuss certain theoretical behavioural aspects
of the derived FIC scheme, and present the results from a simulation study.
Some concluding remarks, some of which pointing 
to future work, are finally provided in Section \ref{section:concludingremarks}.

\section{Estimation and approximations}
\label{section:thefrequencyapproach}

%

We start out investigating the behaviour of the two most common parametric 
estimation procedures, those based on the maximum likelihood method 
and the associated Whittle approximation to the log-likelihood. 
We also give some basics
for nonparametric modelling.


\subsection{Maximum likelihood estimation outside the model}

\def\underliney{\underline y}

Let $\underliney_{n}=(y_1, \ldots, y_n)^\tr$ 
be a collection of $n$ realisations from 
a zero mean stationary Gaussian time series process with spectral distribution 
function $G$ and corresponding spectral density $g$.
Furthermore, let the spectral distribution function $F_{\theta}$ and its corresponding spectral 
density $f_\theta=f(\cdot;\theta)$ index an arbitrary parametric candidate model, where 
$\theta$ belongs to some parameter space $\Theta$ 
of dimension say $p$. The corresponding full log-likelihood is
\beq
	\label{eq:full_loglik}
\ell_n(\theta) = - \frac{n}{2} \log(2 \pi) 
	-\frac{1}{2} \log |\Sigma_{n}(f_{\theta})| 
-\frac{1}{2} \underliney_{n}^{\tr}  \Sigma_{n}(f_{\theta})^{-1} \underliney_{n}, 
\eeq
where $\Sigma_{n}(f_{\theta})$ is the covariance matrix with elements 
\beq
	\notag
	C_{f_{\theta}}(|s - t|)  = 2 \int_{0}^{\pi} \cos(\omega |s - t|) f_{\theta}(\omega) \, \dd \omega 
\eeq
 for $s, t = 1, \ldots, n$.
Since the class of parametric candidate models is not assumed
to necessarily include the true $g$, the maximum likelihood
estimator does not converge to a `true' parameter value.
Instead it converges to the so-called least false parameter value, 
i.e.~$\tildi\theta_n = \argmax_{\theta} \{\ell_{n}(\theta)\} 
\arr_{p} \argmin_{\theta} \{d(g, f_{\theta})\} = \theta_{0}$,  
where 
\beq
\label{definition:discrepancy_measure_d}
\begin{array}{rcl}
d(g, f_{\theta})  
&=&\displaystyle 
   \frac{1}{4\pi} \int_{-\pi}^{\pi} 
\Bigl\{  \frac{ g(\omega) }{ f_{\theta}(\omega) }
   - 1 - \log \frac{ g(\omega) }{f_{\theta}(\omega)} \Bigr\} \, \dd  \omega \\
&=&\displaystyle 
   - \frac{1}{4\pi} \int_{-\pi}^{\pi} 
   \{\log g(\omega)  + 1\} \, \dd \omega- R(G, \theta),
\end{array}
\eeq
and where
\beq
	R(G, \theta) = -\frac{1}{4\pi} \int_{-\pi}^{\pi} 
	\Bigl\{ \log f_{\theta}(\omega) + \frac{ g(\omega) }{ f_{\theta}(\omega) }
	\Bigr\} \, \dd \omega \notag
\eeq
may be referred to as the model specific part, 
see e.g.~\cite{DahlhausWefelmeyer96} for details.
Furthermore, it can be shown that
\beq
\label{eq:lm_large_sample}
\rootn ( \tildi{\theta}_{n} - \theta_{0}) \arr_{d} J_{0}^{-1} U \sim \N_p(0,J_0^{-1}K_0 J_0^{-1}), 
   \quad \textrm{ where } U \sim \N_p(0, K_{0}), 
\eeq 
with $J_{0}$ and $K_{0}$ defined by
\beqn
J_{0} 
&=& J(g, f_{\theta_{0}}) \\
&=& \frac{1}{4 \pi} \int_{-\pi}^{\pi} 
  \Bigl[ \nabla{\Psi}_{\theta_{0}}(\omega) 
   \nabla \Psi_{\theta_{0}}(\omega)^{\tr} g(\omega)
   + \nabla^{2}\Psi_{\theta_{0}}(\omega)\{f_{\theta_{0}}(\omega) 
   - g(\omega)\} \Bigr] 
\frac{1}{ f_{\theta_{0}}(\omega)} \, \dd \omega 
\eeqn
and 
\beqn
K_{0} = K(g, f_{\theta_{0}}) 
= \frac{1}{4 \pi} \int_{-\pi}^{\pi} \nabla{\Psi}_{\theta_{0}}(\omega) 
  \nabla{\Psi}_{\theta_{0}}(\omega)^{\tr} 
  \Bigl\{ \frac{ g(\omega) }{ f_{\theta_{0}}(\omega) } \Bigr\}^{2} \, \dd\omega, 
\eeqn
where $\Psi_{\theta}(\omega) = \log f_{\theta}(\omega)$. and 
$\nabla \Psi_{\theta}(\omega)$ and $\nabla^{2} \Psi_{\theta}(\omega)$
are respectively the vector of partial derivatives and matrix of second order 
partial derivatives with respect to $\theta$, 
see \citet[Theorem 3.3]{DahlhausWefelmeyer96}.
Note that $J_0=K_0$ under model conditions.


\subsection{The Whittle approximation}
The Whittle pseudo-log-likelihood is an approximation 
to the full Gaussian log-likelihood $\ell_{n}$ of (\ref{eq:full_loglik}). 
It was originally suggested by P.~Whittle in the 1950s (cf.~\cite{Whittle53}), 
and is defined as
\beq
\label{definition:whittle_approximation}
\hati{\ell}_{n}(\theta) 
   = - \half n\Big[ \log(2 \pi) 
 + \frac{1}{2 \pi} \int_{-\pi}^{\pi} \log\{2\pi f_\theta(\omega) \} \, \dd \omega
 + \frac{1}{2 \pi} \int_{-\pi}^{\pi} \frac{I_{n}(\omega)}{f_\theta(\omega) } 
	\, \dd \omega \Bigr], 
\eeq
where 
$I_{n}(\omega) = (2\pi n)^{-1} |\sum_{t \le n} y_{t}\exp (i\omega t)|^{2}$
is the periodogram. This approximation 
is close to the full Gaussian log-likelihood in the sense that 
$\ell_{n}(\theta) = \hati{\ell}_{n}(\theta)  + O_p(1)$
uniformly in $f$, see \cite{Coursol82}. 
More important here, however, is that (\ref{definition:whittle_approximation}) 
motivates an alternative estimation procedure, 
namely the Whittle estimator
$\hati{\theta}_{n} = \argmax_{\theta} \{\hati{\ell}_{n}(f_{\theta})\}$.
This  estimator is easier to work with in practice 
(both analytically and numerically) and shares several properties 
with the maximum likelihood estimator. 
In particular $\rootn(\hati{\theta}_{n} - \theta_{0})$
achieves the same limit distribution as in (\ref{eq:lm_large_sample}),
with the same least false parameter value 
$\theta_{0}$ as defined in relation to 
(\ref{definition:discrepancy_measure_d});
see \cite{DahlhausWefelmeyer96} for details. This means that 
in a large-sample perspective, the maximum likelihood 
estimator and the simpler Whittle estimator 
are equally efficient and essentially interchangeable.



\subsection{Nonparametric modelling}

As mentioned in the introduction, we shall use the periodogram 
in (\ref{eq:periodogram}) for nonparametric modelling. Under appropriate 
short memory conditions, it follows from \citet[Theorem 5.5.2]{Brillinger75}  
that $\E\{I_n(\omega)\}=g(\omega)+O(n^{-1})$, i.e.~that the periodogram is 
asymptotically unbiased as an estimator of the spectral density.
We shall thus use 
\begin{equation}
	\label{eq:nonparametric_estimation}
	\hati G_n(\omega) =  \int_{-\pi}^\omega  I_n(u)\,\dd  u, 
\end{equation}
as a canonical estimator for the spectral distribution $G$; for which 
\begin{equation}
	\notag
	\sqrt{n}(\hati{G}_n(\omega) - G(\omega) ) \rightarrow_{d} \N \bigg (0, 4 \pi \int_{-\pi}^{\omega} g(u)^{2} \, \dd u \bigg ), 
\end{equation}
see e.g.~\cite{Taniguchi80}.

\section{Parametric versus nonparametric}
\label{section:parametricornonparametric}

We shall now obtain large-sample approximations
for the focus parameter estimators. These shall
then be used to construct approximate mse formulae 
for each model's estimator of the focus parameter.
When estimated these mses then give the FIC formulae.

\subsection{How to compare parametric and nonparametric models?}

In completely general terms, let $\mu(G)$ be a focus function, 
i.e.~a functional mapping of the spectral distribution $G$ to a scalar value.
This may be estimated parametrically by estimators of the form
$\hatt{\mu}_\pa=\mu(F_{\hatt{\theta}_n})$, or nonparametrically by
$\hatt{\mu}_\np=\mu(\hati{G}_n)$. Other estimators of $\theta$ and $G$ may also be used, however.
Typically, the collection of parametric candidate models does not include the true $G$.
The question is then which model should we 
use  -- parametric or nonparametric -- 
for estimating the focus parameter.

Assume for the nonparametric and each of the parametric candidate models that 
\beqn
\rootn(\hatt{\mu}_{\np} - \mu_\true) \arr_{d} \N(0, v_{\np})
	\quad \textrm{ and } \quad
\rootn(\hatt{\mu}_{\pa} - \mu_{0}) \arr_{d} \N(0, v_{\pa}),
\eeqn
where $\mu_\true=\mu(G)$ is the true value of the focus parameter and
$\mu_{0} = \mu(F_{\theta_{0}})$ is the focus function evaluated 
under the least false parametric model $F_{\theta_{0}}$ as discussed 
in relation to (\ref{definition:discrepancy_measure_d}).
Then, without going into details, the large-sample results above 
motivate the following first-order approximations for the mse 
of the estimated focus parameters:
\beq
\label{eq:FIC_np_pm_mse}
\mse_{\np} = 0^2 + v_{\np}/n  = v_{\np}/n
\quadandquad 
\mse_{\pa} = b^2 + v_{\pa}/n,
\eeq 
where $b = \mu_{0} - \mu_{\true}$. 
The remainder of this section will be used to motivate and obtain 
good estimators for the mean squared errors in (\ref{eq:FIC_np_pm_mse})
with the class of focus paramters of the form $\mu(G; h_0)$ defined in (\ref{eq:focus}),
and the more general $\mu(G;h,H)$ in (\ref{eq:focus2}).

\subsection{Deriving unbiased risk estimates}

In the derivation below, the parametric candidates 
$F_{\theta}$ will be fitted using the Whittle estimator $\hati{\theta}_{n}$ 
as defined in (\ref{definition:whittle_approximation}),
while we will use the canonical periodogoram based estimator in
(\ref{eq:nonparametric_estimation}) for nonparametric estimation of 
the spectral distribution $G$.

Using the Whittle estimator in collaboration with 
(\ref{eq:nonparametric_estimation}) results in a convenient 
simplification of the derivations below; extending 
the arguments to full ML estimation is relatively 
straightforward, using techniques in \cite{DahlhausWefelmeyer96}.
This motivates the following nonparametric and parametric estimators 
for focus parameters $\mu(G; h_0)$ on the form of (\ref{eq:focus}):
\beqn
\hati\mu_{\np} 
= \int_{-\pi}^{\pi} h_0(\omega) I_{n}(\omega) \, \dd \omega 
= \frac{1}{n} \underliney_{n}^{\tr} \Sigma_{n}(h_0) \underliney_{n}
   \quadandquad 
\hati\mu_{\pa} = \int_{-\pi}^{\pi} h_0(\omega) 
   f_{\hati{\theta}_{n}}(\omega)\,\dd \omega,
\eeqn
where $\Sigma_{n}(h_0)$ is a $n\times n$-dimensional symmetric Toeplitz matrix, having elements of the general form
\beq
	\sigma_{n,s,t}(h_0) = \int_{-\pi}^{\pi} \cos(\omega |s - t|) h_0(\omega) \, \dd \omega. \notag
\eeq
for $s,t=1,\ldots,n$. The following proposition establishes the joint limit distribution 
for the estimators above (suitably normalised), which
in turn will be used to obtain good approximations for 
their respective mean squared errors. 




\begin{proposition} 
\label{proposition:propone}
Let $y_{1}, \ldots, y_{n}$ be realisations from 
a stationary Gaussian time series model with 
spectral density $g$ assumed to be uniformly bounded 
away from both zero and infinity. 
Suppose $|h_0|$ is bounded in $\omega$, that $f_{\theta}$ is two times differentiable with respect 
to $\theta$, and that $f_{\theta}$ and these derivatives, $\nabla f_{\theta}$ and $\nabla^{2} f_{\theta}$,  are continuous and uniformly bounded in 
both $\omega$ and $\theta$ in a neighbourhood 
of the least false parameter value $\theta_{0}$ as defined in 
(\ref{definition:discrepancy_measure_d}) above. Then 
\beq
\label{limit:mu_pa_and_mu_np}
\left(
\begin{array}{c}
\rootn(\hati{\mu}_{\np} - \mu_{\true}) \\
\rootn(\hati{\mu}_{\pa} - \mu_{0})
\end{array} \right)
\arr_d
\left( 
\begin{array}{c}
X_0 \\
c_0^{\tr} J(g, f_{\theta_{0}})^{-1} U
\end{array} \right)
\sim \N_{2} \left (
\left ( \begin{array}{c} 0 \\ 0 \end{array} \right ), 
\left ( \begin{array}{cc}
v_{\np}  & v_{c} \\
v_{c}      &  v_{\pa}
\end{array} \right)\right ),
\eeq
where
\beqn
v_{\np} = 4 \pi \int_{-\pi}^{\pi} \{  h_0(\omega) g(\omega) \}^{2} \, \dd \omega
\quadandquad 
v_{\pa} = c_0^{\tr} J(g, f_{\theta_{0}})^{-1} K(g, f_{\theta_{0}})  
J(g, f_{\theta_{0}})^{-1} c_0,
\eeqn
with $J$ and $K$ as defined below (\ref{eq:lm_large_sample}), 
and $v_{c} = c_0^{\tr} J(g, f_{\theta_{0}})^{-1} d_0$, 
where the $c_0$ is the partial derivative of $\mu(F_{\theta_0};h)$ with respect to $\theta$, i.e.~$c_0 = \nabla \mu(F_{\theta_{0}};h)= \int_{-\pi}^\pi h_0(\omega)\nabla f_{\theta_0}(\omega)\, \dd \omega$ 
 and 
$$d_0 = \cov(X, U) =  \int_{-\pi}^{\pi} 
\frac{\nabla f_{\theta_{0}}(\omega)  h_0(\omega) 
   g(\omega)^{2}}{f_{\theta_{0}}(\omega)^{2}} \, \dd \omega. $$
\end{proposition}

\begin{proof}

It follows from the results in \cite[Ch.~2]{Dzhaparidze86} that
$\hati\theta_n - \theta_0 = J(g, f_{\theta_{0}})^{-1} U_{n} + o_p(1/\rootn)$, 
where $U_n= \nabla \hati{\ell}_{n}(f_{\theta_{0}}) $ and 
\beqn
	U_{n}
	= - \half \{ \Tr ( \Sigma_{n}(\nabla \Psi_{\theta_{0}}) ) - 
   \underliney_{n}^{\tr} \Sigma_{n}( \nabla \Psi_{\theta_{0}}/ f_{\theta_{0}}) 
   \underliney_{n}  \}, 
\eeqn
where 
$\Psi_{\theta_{0}} = \log f_{\theta_{0}}$ and $\nabla \Psi_{\theta_{0}}$ is 
the vector of its partial derivatives. 
As a consequence, a Taylor expansion motivated by the standard delta method gives
$\hati{\mu}_\pa - \mu_0 = c_0^{\tr}J(g, f_{\theta_{0}})^{-1} U_{n} + o_p(1/\rootn)$.
Since $\sqrt{n}U_n \arr_d U$ by the assumptions of the proposition \citep{Dzhaparidze86}, the parametric part of the result holds.
In addition 
$$X_n= (\hati{\mu}_{\np} - \mu_{\true}) 
= \frac{1}{n}\underliney_{n}^{\tr} \Sigma_{n}(h_0) 
\underliney_{n} - \mu_\true,$$
which can be shown, by a modified version of the argument leading to 
the limit distribution of $U_n$,  
to have the property that $\sqrt{n}X_n \arr_d X_0 \sim \N(0,v_\np)$. 
This proves the nonparametric part of the result.
We finally need to show that these convergence results hold jointly.
Since the two drivers in the derivation of the limit distribution, 
$\underliney_{n}^{\tr} \Sigma_{n}(h_0) \underliney_{n} /n$ 
and $U_{n}$, are quadratic forms, the joint limit distribution 
is readily obtainable by a Cram{\'e}r--Wold type of argument.
To see how, let $a$ be a vector in $\R^{2}$ to be used in the Cram{\'e}r--Wold argument, and
define $$
	\Lambda_{n} 
	= a_{1} \sqrt{n} X_{n} + a_{2} \sqrt{n} U_{n}
	=
	\frac{1}{\sqrt{n}} \underliney_{n}^{\tr}\Sigma_{n}(a_{1} h_0 
	+  a_{2} \nabla \Psi_{\theta_{0}}/ f_{\theta_{0}}) \underliney_{n} 
	+ \gamma_n
$$
where $\gamma_n = \sqrt{n} \{ a_1\mu_\true - a_2 \Tr ( \Sigma_{n}(\nabla \Psi_{\theta_{0}}) )/2\}$. 
The $\gamma_n$ cancels out the mean, here, such that $\Lambda_n$ has mean zero. 
This is once again just a quadratic form, hence, $\Lambda_n$ is normal under the assumptions of the proposition; 
see \cite{Dzhaparidze86} or 
\cite{HermansenHjort14a} for derivations of a similar type.
The proof is completed by observing that by \citet[Lemma A.5]{DahlhausWefelmeyer96}, the covariances 
take the relevant form 
\beqn 
\cov(X_{n}, U_{n}) 
= \frac{2}{n} \Tr \{  \Sigma_{n}(h_0) \Sigma_{n}(g) 
\Sigma_{n}(\nabla \Psi_{\theta}/ f_{\theta}) \Sigma_{n}(g)  \}
\arr 
\int_{-\pi}^{\pi} 
\frac{\nabla f_{\theta_{0}}(\omega)h_0(\omega)g(\omega)^{2}}{f_{\theta_0}(\omega)^{2}} 
   \, \dd \omega.
\eeqn 
\end{proof}

We next extend the above proposition to the more general class of
We next extend the above proposition to the more general class of
focus parameters $\mu(G;h,H)$ in (\ref{eq:focus2}), being a continuously differentiable function of a finite number of the $\mu(G;h_0)$ functions. 
The nonparametric and parametric estimators for this class take the form 
$$\hati\mu_\np
   =H\bigl(n^{-1}\underliney_{n}^{\tr} \Sigma_{n}(h_1) \underliney_{n},
   \ldots,n^{-1}\underliney_n^\tr\Sigma_n(h_k)\underliney_n\bigr) $$
and 
$$\hati\mu_\pa
   =H\Bigl(\int_{-\pi}^\pi h_1(\omega)f(\omega;\hati\theta_n)\,\dd\omega,
   \ldots,\int_{-\pi}^\pi h_k(\omega)f(\omega;\hati\theta_n)\,\dd\omega\Bigr). $$ 

\begin{proposition}
\label{proposition:proptwo}
Under the conditions of Proposition \ref{proposition:propone} 
the focus parameters $\mu(G;h, H)$ in (\ref{eq:focus2}),
with estimators and estimands as above, 
fulfils 
\begin{align}
\label{limit:mu_pa_and_mu_np2}
\left(
\begin{array}{c}
\rootn(\hati{\mu}_{\np} - \mu_{\true}) \\
\rootn(\hati{\mu}_{\pa} - \mu_{0})
\end{array} \right)
\arr_d
\left( 
\begin{array}{c}
\nabla{H}_\np X \\
\nabla{H}_\pa c^{\tr} J(g, f_{\theta_{0}})^{-1} U
\end{array} \right)
\sim \N_{2} \left (
\left ( \begin{array}{c} 0 \\ 0 \end{array} \right ), 
\left ( \begin{array}{cc}
v_{\np}  & v_{c} \\
v_{c}      &  v_{\pa}
\end{array} \right)\right ),
\end{align}
where
\beqn
v_{\np} &=&  \nabla H_\np \{4 \pi \int_{-\pi}^{\pi} \{  h(\omega) g(\omega) \}^{2} \, \dd \omega\} \nabla H_\np^{\tr}
\quadandquad \\
v_{\pa} &=& \nabla H_\pa c^{\tr} J(g, f_{\theta_{0}})^{-1} K(g, f_{\theta_{0}})  
J(g, f_{\theta_{0}})^{-1} c \nabla H_\pa^{\tr},
\eeqn
and $v_{c} = \nabla H_\pa c^{\tr} J(g, f_{\theta_{0}})^{-1} d\, \nabla H^{\tr}_\np$,
where
$\nabla H_\np$ and $\nabla H_\pa$ are the gradients of $H$ evaluated at respectively $(\mu(G;h_1),\ldots,\mu(G;h_k))$ and
$(\mu(F_{\theta_0};h_1),\ldots,\mu(F_{\theta_0};h_k))$, 
$c$ is the $k\times p$-dimensional matrix with rows given by $\nabla \mu(F_{\theta_0};h_j), j=1,\ldots,k$
and 
$$
d=\cov(X,U)
=\int_{-\pi}^\pi \frac{\nabla f_{\theta_0}(\omega)h(\omega)g(\omega)^2}{f_{\theta_0}(\omega)^2}\, \dd \omega.
$$
\end{proposition}

\begin{proof}
By Propostion \ref{proposition:propone}, we see that 
(\ref{limit:mu_pa_and_mu_np}) holds for each $\mu(G;h_j)$. 
Let now $X_{n,j} =\frac{1}{n} \underliney_{n}^{\tr} \Sigma_{n}(h_j) \underliney_{n} - \mu_\true$ for $j=1,\ldots,k$.
By extending the Cramér--Wold argument in Propostion \ref{proposition:propone} to all of
$X_{n,1},\ldots, X_{n,k},U_n$, we see that there is joint convergence for all these. 
The standard (multivariate) delta method then completes the proof. 
\end{proof}

\begin{remark}
From the underlying structure of the proof of Propositions \ref{proposition:propone} and \ref{proposition:proptwo}, and the arguments (of e.g.~\cite{DahlhausWefelmeyer96} or \cite{Dzhaparidze86}) used to  show that the Whittle estimator has the same large-sample properties 
as the maximum likelihood estimator, it is clear that the conclusions of the two propositions
stays true if we replace Whittle with full maximum likelihood estimation. 
\end{remark}

The nonparametric estimator is by construction unbiased in the 
limit; an estimate for the risk is therefore 
easily obtained from the variance formula above. 
For the parametric candidate, we need in addition an unbiased 
estimate for the squared bias. Following \cite{JullumHjort15} 
we start with $\hati{b} =  \hati{\mu}_{\pa} - \hati{\mu}_{\np}$ 
as an initial estimate for $b = \mu_{0} - \mu_{\true}$. 
Since it follows from (\ref{limit:mu_pa_and_mu_np}) that 
$\rootn ( \hati{b} - b) \arr_{d} c^{t} J^{-1} U - X \sim \N(0, \kappa)$, 
where $\kappa = v_{\pa} + v_{\np} - 2v_{c}$, we have
$\E \, \hati{b}^{2} \approx b^{2} + \kappa/n + o(1/n)$.
This leads to mse estimators of the form
\begin{align}
\label{eq:FIC}
\begin{aligned}
\fic_{\np} &= \hati{\mse}_{\np}  = \hati{v}_{\np}/n, \\
\fic_{\pa} &= \hati{\mse}_{\pa}  = \hati{\bsq} + \hati{v}_{\pa}/n 
	= \max(0, \hati{b}^{2} - \hati{\kappa}/n)+ \hati{v}_{\pa}/n.
\end{aligned}
\end{align}
For the most general focus parameter formulation in (\ref{eq:focus2}), the variance and covariance estimators take the form 
\beqn
	\hati{v}_{\np} & = & \nabla \hati{H}_\np \{2 \pi \int_{-\pi}^{\pi}   
	h(\omega)^{2} I_{n}(\omega)^2  \, \dd \omega\} \nabla \hati{H}_\np^{\tr}, \textrm{ and }  \\
	\hati{v}_{\pa} & = & \nabla \hati{H}_\pa\hati{c}^{\; \tr} J(I_{n}, f_{\hati{\theta}_{n}})^{-1} 
	K(I_{n} / \sqrt{2}, f_{\hati{\theta}_{n}})  
	J(I_n, f_{\hati{\theta}_{n}})^{-1} \hati{c} \; \nabla \hati{H}_\pa  
	,
\eeqn
where $\hati{c} = (\nabla \mu(F_{\hati{\theta}_{n}};h_k),\ldots,\nabla \mu(F_{\hati{\theta}_{n}};h_k))^{\tr}$, 
$\nabla \hati{H}_\np$ and $\nabla \hati{H}_\pa$ are the gradients of $H$ evaluated at respectively $(\mu(\hati{G}_n;h_1),\ldots,\mu(\hati{G}_n;h_k))$ and
$(\mu(F_{\hati{\theta}_{n}};h_1),\ldots,\mu(F_{\hati{\theta}_{n}};h_k))$, 
and $J$ and $K$ are 
as defined in relation to $(\ref{eq:lm_large_sample})$ -- using  
$I_n(w)^2/2$ as the canonical nonparametric unbiased estimator 
for $g(w)^{2}$. These are all consistent according to \cite{Taniguchi80,Deo00}.


With FIC scores as above, representing clear-cut estimates of the risk of the nonparametric 
and parametric models' estimators of $\mu$, our model selection strategy turns 
out as follows: Compute the FIC score for each candidate model,
rank them accordingly, and select the model and estimator 
associated with the smallest FIC score. The same $\fic_{\pa}$ 
formula (with different estimates and quantities) 
is used for all of the possibly $m$ different parametric 
candidate models for simultaneous selection among the $m + 1$ 
models. This is perfectly fine as the $\fic_{\pa}$ 
formula does not depend on the other parametric models.

 Although we have concentrated on focus functions $\mu(G;h)$ and 
$\mu(G;h,H)$ given by (\ref{eq:focus}-\ref{eq:focus2}), our focused model selection 
strategy applies also to more general focus parameters, as long as 
joint limit distributions like (\ref{limit:mu_pa_and_mu_np}) and 
(\ref{limit:mu_pa_and_mu_np2}) may be proven. 
In completely general terms, our results may be generalised 
to focus parameters of the form $\mu=T(G)$ for well-behaved functionals 
$T$ mapping the spectral distribution $G$ to a scalar value.
The type of smoothness required for $T$ is in fact that the 
functional is so-called Hadamard differentiable at $G$ and $F_{\theta_0}$,
see e.g.~\citet[Theorem 20.8]{vanderVaart98} for further details. 
This allows us, for instance, to handle focus parameters involving
quantiles of the spectral distribution $G$.
It is also possible to extend the arguments 
to other parametric estimation procedures, especially if they 
are derived as minimisers of the empirical analogue of 
$\argmin\{R(G,\theta)\}$ for $R$ the model specific 
part of possibly different divergence measure than in 
(\ref{definition:discrepancy_measure_d}), 
see \cite{DahlhausWefelmeyer96} and \cite{Taniguchi80}
for alternatives.

%
%

\section{Models with trends}
\label{section:detrending}

So far we have only considered stationary time series with 
mean zero. In real applications, this is often an unrealistic 
assumption to make. Even if the series is stationary, the underlying 
mean is rarely exactly zero; the common solution 
in such cases is to detrend the series.
In time series modelling, detrending usually refers to the act 
of removing an estimated or deterministic trend from the observed 
series before the main analysis. This may be a complex function 
of time and covariates including seasonal effects, or be as simple as 
subtracting the arithmetic mean. A common approach is to work 
with the detrended series, which we will denote by $\hatt{y}_{t}$, and then
analyse this series using models for stationary time series,
without factoring in the extra estimation uncertainty 
involved in the detrending. This is often unproblematic, but 
even the innocent action of subtracting the mean may have 
unforeseen consequences (typically for the so-called second 
order properties). \cite{HermansenHjort14a} shows that such 
a simple operation alter the underlying motivation and
interpretation of the AIC for stationary Gaussian time series.
Thus, special care is required for such an operation.

Suppose the observed series is generated by the model 
\beq
\label{eq:time_series_with_trend}
Y_{t} = m(x_{t}, \beta) + \eps_{t}, 
\eeq 
where the $x_{t}$ are $p$-dimensional covariates, 
the $m$ is of known parametric structure, and 
$\{\eps_{t}\}$ is a zero mean stationary Gaussian time 
series process with spectral distribution 
function $G$ and corresponding density $g$. Assume 
further that we are able to estimate $\beta$ by a 
suitable $\hatt{\beta}_{n}$ with reasonable precision.
The question is then whether the results of 
Section \ref{section:parametricornonparametric}
are still valid also with detrended data, such that we may still use the same FIC formulae.

\begin{proposition}
\label{proposition:propthree}
Suppose the spectral densities $g$ and $f_{\theta}$ and 
function $h$ satisfy the conditions of 
Proposition \ref{proposition:propone}, and that the 
assumed trend $m$ and corresponding estimator 
$\hatt{\beta}$ for the unknown $\beta$
are such that $\rootn(\hatt\beta_n - \beta) = O_p(1)$. 
Assume further that in a neighbourhood of $\beta$ we have
\beqn
m(x, \hatt{\beta}_{n}) = m(x, \beta) 
   + \nabla m(x, \beta)^\tr (\hatt\beta_n - \beta) + r_{n}(x),
\eeqn
with $\max_{i} |r_{n}(x_{i})| = o_p(1/\sqrt{n})$ and $|\nabla m(x, \beta)|$
bounded in $x$.
Then the conclusions of Proposition \ref{proposition:propone} 
are still true if we replace $y_{t}$ with the detrended 
$\hatt{y}_{t} = y_{t} - m(x_{t}, \hatt{\beta}_{n})$.
\end{proposition}


\begin{proof}
We will show that the result follows as a corollary from 
certain general results regarding limit behaviour of 
quadratic forms from \citet[Section 3]{HermansenHjort14b}.
 
The argument is structured similarly to that of 
Proposition \ref{proposition:propone} and is built 
around a Cram{\'e}r--Wold type of argument. Observe 
that if we replace $y_t$ with the detrended 
$\hatt{y}_t = y_{t} - m(x_{t}, \hatt{\beta}_{n})$, 
we now have  
$\hatt{X}_{n}= ( \hatt{\underliney}_{n}^{\tr} \Sigma_{n}(h_{0}) 
   \hatt{\underliney}_{n} - \mu_{\true})$
and similarly 
\beqn
\hatt{U}_n
= - \half \{ \Tr ( \Sigma_{n}(\nabla \Psi_{\theta_{0}}) ) - 
  \hatt{\underliney}_{n}^{\tr} \Sigma_{n}
   ( \nabla \Psi_{\theta_{0}}/ f_{\theta_0}) 
   \hatt{\underliney}_{n}  \},
\eeqn
where $\hatt{\underliney}_{n}=(\hatt{y}_1, \ldots, \hatt{y}_n)^\tr$. Again, for any $a=(a_1,a_2)$ in $\R^2$, 
we now have 
$$
\hatt{\Lambda}_{n} = a_{1}\sqrt{n} \hatt{X}_{n} 
+ a_{2}\sqrt{n} \hatt{U}_{n} 
= \hatt{\underliney}_{n}^{\tr}\Sigma_{n}(a_{1} h_{0} 
+  a_{2} \nabla \Psi_{\theta_{0}}/ f_{\theta_{0}}) \hatt{\underliney}_{n}/\rootn + \gamma_n,
$$
with $\gamma_{n}$ as in the proof of Proposition \ref{proposition:propone}.
Then,  according to Proposition 3.1 of \citet{HermansenHjort14b}, 
$$ \hatt{\Lambda}_{n} - \Lambda_{n} = o_p(n^{-1/2}) $$
where $\Lambda_{n} = \underline{\eps}_{n}^{\tr}\Sigma_{n}(a_{1} h_{0} 
+  a_{2} \nabla \Psi_{\theta_{0}}/ f_{\theta_{0}}) \underline{\eps}_{n} /\sqrt{n} + \gamma_{n}$, 
where $\underline{\eps}_{n}=(\eps_1,\ldots,\eps_n)^{\tr}$ has elements corresponding to (\ref{eq:time_series_with_trend}).
Since the limit behaviour of $\Lambda_{n}$ is what 
defines the limit distribution in Proposition \ref{proposition:propone}, the 
argument is essentially complete.
\end{proof}

The above proposition may also be extended to the focus parameter in \eqref{eq:focus2}, as 
handled in Proposition \ref{proposition:proptwo}.
Traditionally, the least squares estimator has been the canonical method 
for estimating $\beta$ in models of the form of (\ref{eq:time_series_with_trend}).
 As an illustration, consider the linear regression model 
with dependent errors where $Y_{t} = x_{t}^{\tr} \beta + \eps_{t}$, for $p$-dimensional covariates $x_{t}$, 
and where $\{\eps_{t}\}$ is a zero mean stationary Gaussian time 
series process with spectral density $g$. On matrix form this yields
$\underline{y}_{n} = X \beta + \underline{\eps}_{n}$, where $X$ is 
the related $n \times p$-dimensional design matrix. The ordinary least squares 
estimate for $\beta$ is then given by 
$\hatt{\beta}_{n} =  (X^{\tr}  X)^{-1}X^{\tr} \underline{y}_{n}$. Then, in order for 
$\hatt{\beta}_{n}$ to satisfy the conditions of Proposition \ref{proposition:propthree},
it is sufficient that $n \Var ( \hatt{\beta}_{n} )
= n (X^{\tr} X)^{-1} X^{\tr} \Sigma_{n}(g) X (X^{\tr} X)^{-1} = o(1)$,
which is clearly satisfied if 
$X^{\tr} X/n \rightarrow_{p} Q_1$ and $X^{\tr} \Sigma(g) X/n \rightarrow_{p} Q_2$, 
as $n$ approaches infinity, where $Q_1$ and $Q_2$ are both 
finite positive definite matrices. These are the standard assumptions needed 
to ensure consistency of both standard and generalised least squares for models
with correlated errors.

\section{Average focused information criterion}
\label{section:afic}

We have so far concentrated on inference for a single 
focus parameter $\mu$. A natural generalisation of this
is to consider several focus parameters joinly, 
say correlations of orders 1 to 5. The FIC machinery can easily be lifted to such a situation,
involving a weighted average of FIC scores, the AFIC, with weights reflecting importance dictated  by the statistician. 

Suppose in general terms that estimands $\mu(u)$ 
are under consideration, for $u$ in some index set.
For each of these we have the nonparametric 
$\hati\mu_\np(u)$ and one or more parametric estimators
$\hati\mu_\pa(u)$. These typically have versions
of Propositions \ref{proposition:propone} or \ref{proposition:proptwo}, 
leading as per (\ref{eq:FIC_np_pm_mse}) to 
$$\mse_\np(u)=0^2+v_\np(u)
   \quadandquad
  \mse_\pa(u)=b(u)^2+v_\pa(u), 
  $$
with $b(u)=\mu_0(u)-\mu_\true(u)$. These mean squared 
errors can then be combined, via some suitable 
cumulative weight function $W(u)$, to 
$$\risk_\np= \int v_\np(u)\, \dd W(u)
   \quadandquad
  \risk_\pa=\int\{b(u)^2+v_\pa(u)\}\, \dd W(u) $$ 
Here $\dd W(\cdot)$ is meant to reflect the relative 
importance of the different $\mu(u)$, and should 
stem from the statistician's judgement and the 
actual context. Based on the data we may now form the following
natural estimates of these risk quantities:
\begin{align}
\label{eq:afic}
\begin{aligned}
\afic_\np&= \int\hatt v_\pa(u)\, \dd W(u), \\
\afic_\pa &= \int \bigl[\max\{\hatt b(u)^2-\hatt\kappa(u)/n\} 
   + \hatt v_\pa(u)\bigr]\, \dd W(u).
\end{aligned}
\end{align}
This operation also needs the covariances $v_c(u)$, 
as $\hatt\kappa(u)$ is to be constructed as the 
natural estimator of $\kappa(u)=v_\pa(u)+v_\pa(u)-2v_c(u)$. 

The AFIC scheme (\ref{eq:afic}) can be used in a variety
of circumstances.
A typical application may
involve assessing models for estimating a threshold
probability $\Pr\{Y_{n+1}\ge a\}$ over a set of many $a$,
again with a weight function $w(a)$ indicating 
relative importance. Another attractive application 
is for the task of estimating correlations $\corr(h)$
for lags $h=1,2,3,\ldots$, perhaps with a decreasing $w(h)$.
The AFIC method may similarly be applied for 
comparing the popular autorcorrelation function, such as 
{\tt acf} in the statistical software package {\tt R} \citep{R}, 
with potentially more accurate parametric alternatives.


%

\section{Performance}
\label{section:performance}

In the present section we will discuss some behavioural aspects of the derived FIC methodology.
First we present some theoretical consequences of using our new FIC construction for model selection.
Then we discuss some issues related to the 
more practical performance of this criterion,
and illustrate some of these in a simulation study. 
The goal is not to conduct a broad simulation 
based investigation, but rather show the potential of having a criterion 
for selecting among parametric models and a nonparametric alternative 
in a simple proof of concept type of illustration.

\subsection{FIC under model conditions}

Although we have been working outside specific parametric model 
conditions when deriving the FIC (and AFIC) above, it is natural to 
ask how the criteria selects when a parametric model is indeed correct. 
Consider however first the case where a specific parametric candidate 
model is incorrect and have bias $b\neq 0$. From the structure of the FIC formulae in (\ref{eq:FIC}) and the consistency of the involved variance and covariance estimators, we see that $\text{FIC}_\np =o_p(1)$, while $\text{FIC}_\pa = O_p(1)+o_p(1)=O_p(1)$. I.e.~the squared bias term dominates completely, and the probability that the FIC will select this particular parametric model will tend to  0 as $n \arr \infty$. If all the parametric candidate models are biased in this sense, then the FIC will eventually prefer the nonparametric model when the sample size increases.

Going more into detail, it is seen from the FIC formulae in (\ref{eq:FIC}) that the FIC prefers a specific parametric model over the nonparametric whenever
$$\max(\hati b^2-\hati\kappa/n,0) + n^{-1}\hati v_\pa 
   \leq  n^{-1}\hati v_\np. $$ 
Whenever $\hati v_\np \geq \hati v_\pa$, this is seen to be equivalent to 
$$Z_n \leq 2,$$ 
where $Z_n = (n\hati b^2)/(\hati v_\np - \hati v_c)$.

It turns out that under model conditions, we have $v_c=v_\pa$. This is rather straightforward to see by investigating the forms of $v_c$ and $v_\pa$ involved in Proposition \ref{proposition:proptwo}, in addition to the forms of $K_0$ and $J_0$. Inserting $g=f_{\theta_0}$ in these formulae reveals that $K_0=J_0$, $\nabla H_\np = \nabla H_\pa$ and $c=d$ and thereby $v_c=v_\pa$. 
Now, due to the consistency, we have $\hati{v}_\np - \hati{v}_c \arr_p v_\np - v_\pa$. Further, the limit distribution result of $\rootn(\hati b - b)$ given above \eqref{eq:FIC} ensures that
$Z_n \arr_d \chi_1^2$, with $\chi^2_1$ a chi-squared distributed variable with one degree of freedom.
That is, the limiting probability that the parametric model will be selected over the nonparametric when it is indeed true is
$\Pr\{Z_n \le 2\} \arr \Pr\{\chi_1^2 \le 2\} \approx 0.843$. 
Thus, if one of the parametric candidate models is correct, and the others have biases $b\neq 0$, then, for sufficiently large samples, the first parametric model and estimator will be selected with a probability tending to 84.3\%, while the nonparametric will be selected in the other 15.7\% proportion.

\subsection{FIC in practice}

\begin{figure}[ht]
\centering
\includegraphics[width=0.90\textwidth]{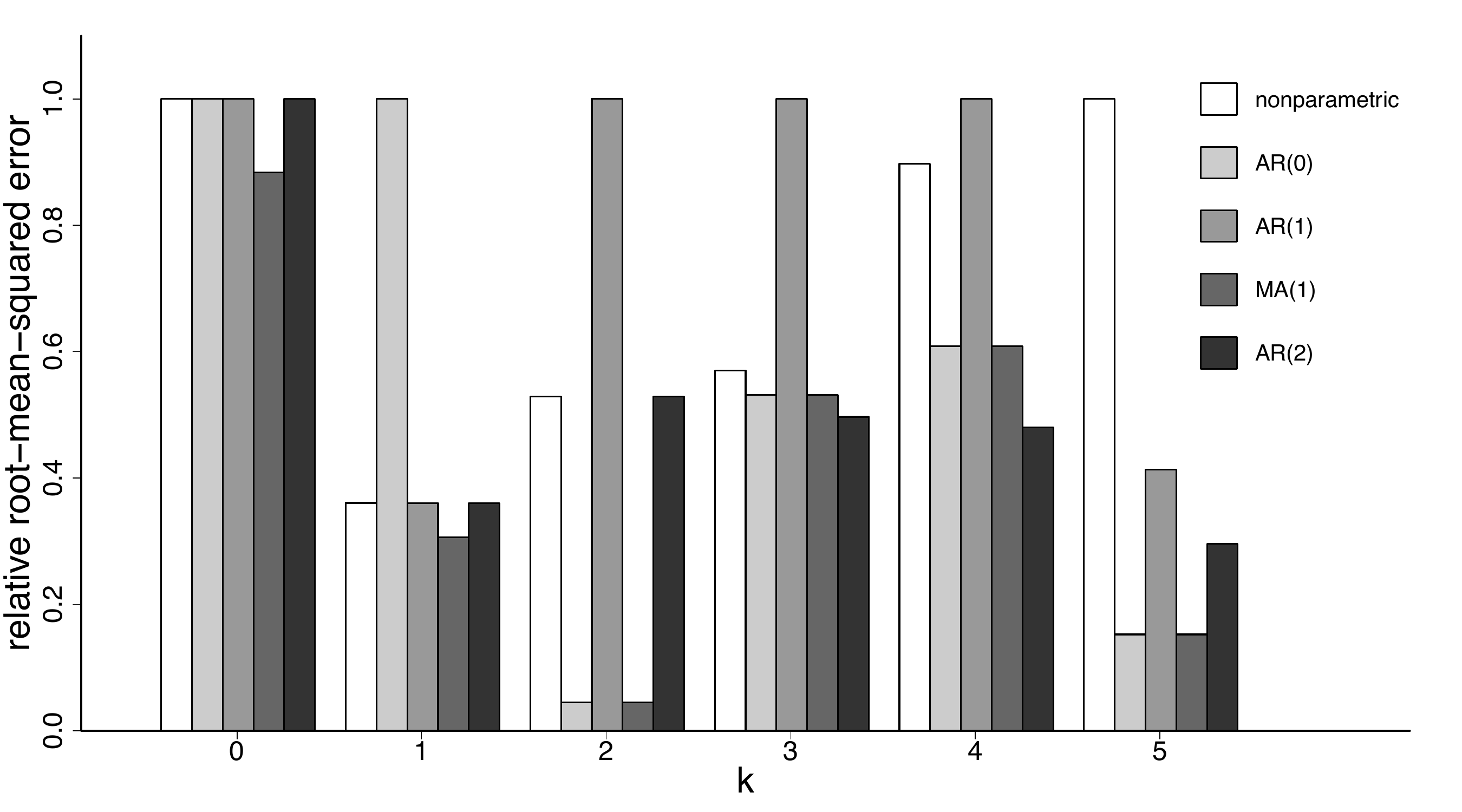}
\caption{
Relative root-mse for each candidate model fitted to the six
focus parameters $\mu_{k} = C(k)$, for $k = 0, \ldots, 5$. 
The root-mse is computed based on 5000 simulated AR(2) series 
of length $n = 100$, with $\sigma = 1.0$ and $\rho = (0.7, -0.6)$,
For ease of comparison we have scaled the root-mse to the unit interval.}
\label{figure:limit_mse}
\end{figure}

\begin{figure}[ht]
\centering
\includegraphics[width=0.90\textwidth]{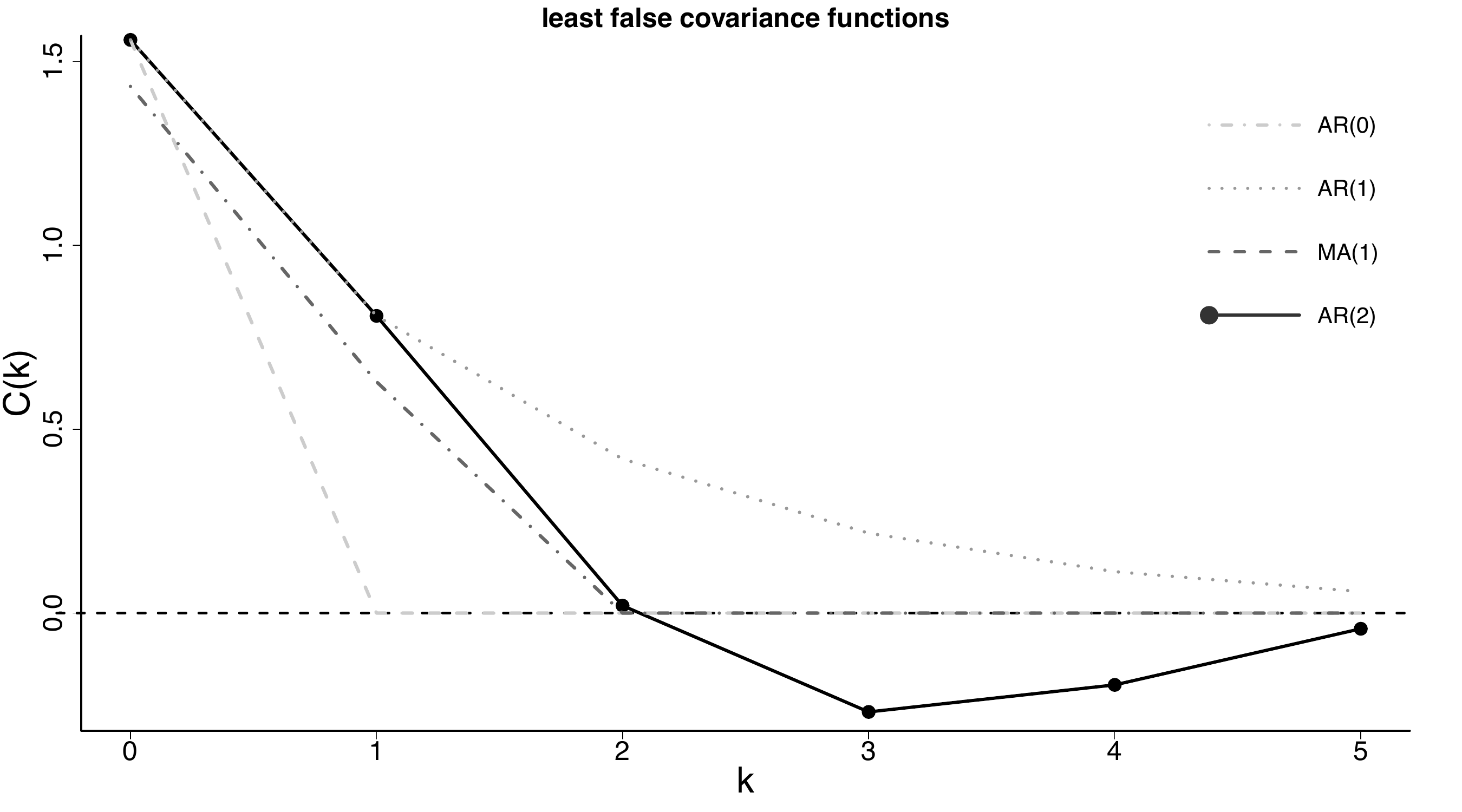}
\caption{
The five least false covariance functions under 
the assumption that the true model is an autoregressive model 
specified by the parameters $\sigma = 1.0$ and 
$\rho = (0.7, -0.6)$.}
\label{figure:least_false}
\end{figure}

Figure \ref{figure:limit_mse} shows the relative root-mse for 
estimating the focus parameter 
\begin{equation}
	\label{eq:mu_k}
	\mu_{k} = \mu(G; h_{k}) = \int_{-\pi}^{\pi} 
	\cos (\omega k) g(\omega) \, \dd \omega = C_{g}(k), \quad \textrm{ for } k = 1, \ldots, 5, 
\end{equation}
based on the following five candidates models: 
the independence model (autoregressive of order zero); 
the autoregressive of orders one and two; 
the moving average of order one; 
and finally the nonparametric one, where nothing
more is assumed than saying that the series is stationary with
a finite variance. 
The true model is an autoregressive 
model specified by the parameters 
$\rho = (0.7, -0.6)$ and $\sigma = 1.0$. This means that all but two, the autoregressive 
model of order two and the nonparametric model, are misspecified. 
The corresponding least false covariance estimates are plotted in 
Figure \ref{figure:least_false}. In the simulation study, we have used 
$B = 5000$ repetitions of length $n  =100$ to compute the 
actual relative root-mse values for each candidate.
Note that since we have included the true model among 
our candidates, nonparametric estimation is never the 
optimal choice; it is however often close and it is the second 
best choice for lags 1 and 3. For lags 2 and 5, where the true values 
are close to zero, the simpler models, like AR(0) and
MA(1), are highly successful, achieving reasonably  
low bias and also low variance. 

\begin{figure}[ht]
\centering
\includegraphics[width=0.90\textwidth]{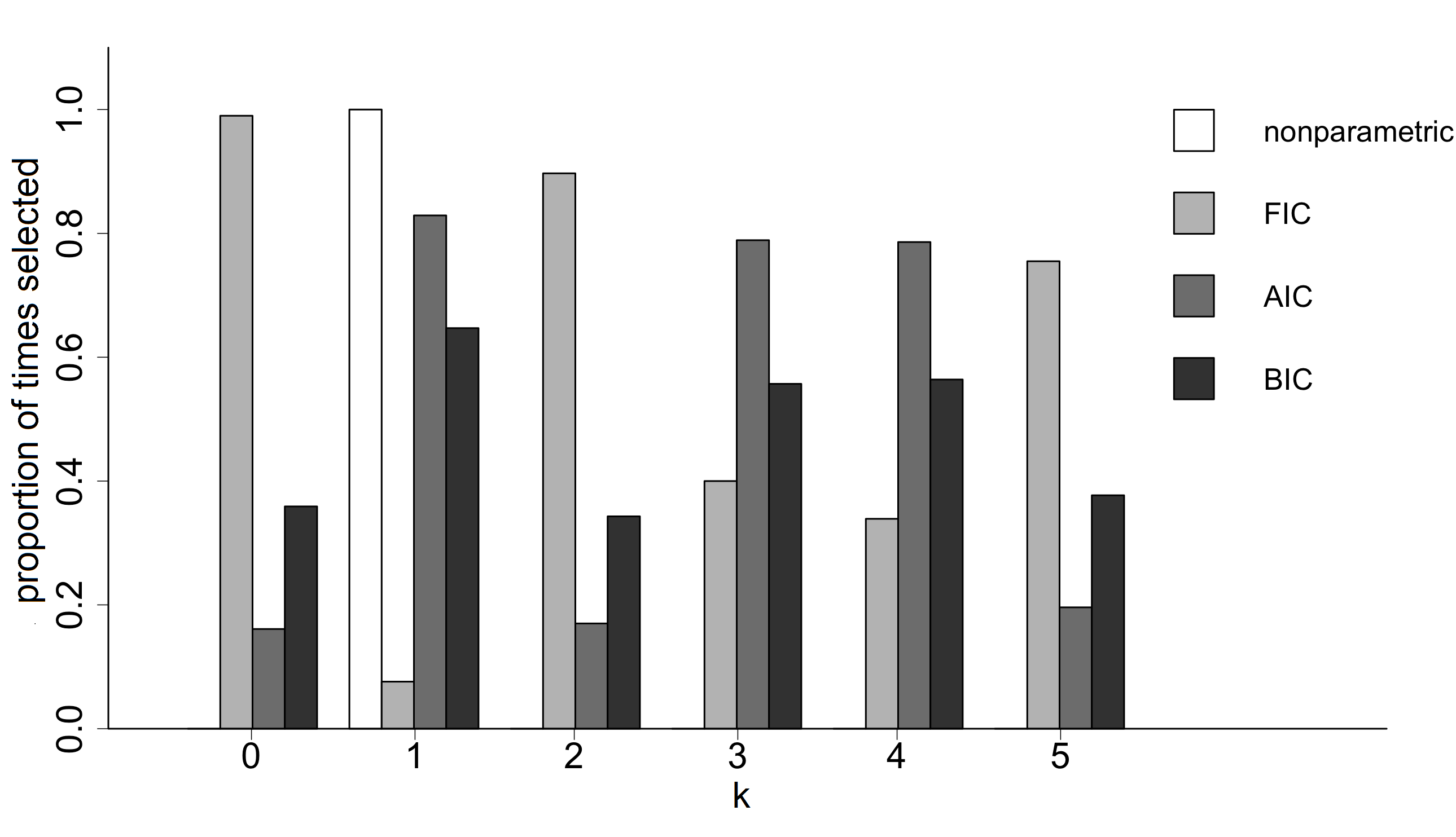}
\caption{
The proportion for which the different criteria selects the model with the theoretical lowest 
root-mean-squared error. The model-selectors are always nonparametric, FIC, AIC and BIC. 
The results are based on 5000 simulated series. 
}
\label{figure:simulation1}
\end{figure}

In Figure \ref{figure:simulation1} and  \ref{figure:simulation2} 
we further investigate the performance of the FIC. Here, we 
compare our FIC machinery
with three other model selection strategies, (i) to always 
use the nonparametric model, (ii) select the best parametric model 
according to the AIC and (iii) the parametric model selected by the BIC.
Note that the AIC and BIC tools do not work for the nonparametric 
model, since there is no likelihood function.
In Figure \ref{figure:simulation1} we have counted how many times 
each criterion selects the 
model that obtains the smallest root-mse value, for each focus parameter $\mu_{k}$ as defined in (\ref{eq:mu_k}).
Figure \ref{figure:simulation2} contains the corresponding attained root-mse values. 
Note that for lag 1 the theoretical root-mse for the autoregressive models
are, for all practical purposes, equal to that obtained by 
the nonparametric model. 
In all other cases, the nonparametric model 
has a root-mse larger than the optimal model.

In this illustration, the FIC behaves more or less 
as intended, by selecting (on average) the models that 
produces the smallest risk. The amount of evidence is by no
means conclusive, but it indicates that the FIC machinery 
has a real potential.

 \begin{figure}[ht]
\centering
\includegraphics[width=0.90\textwidth]{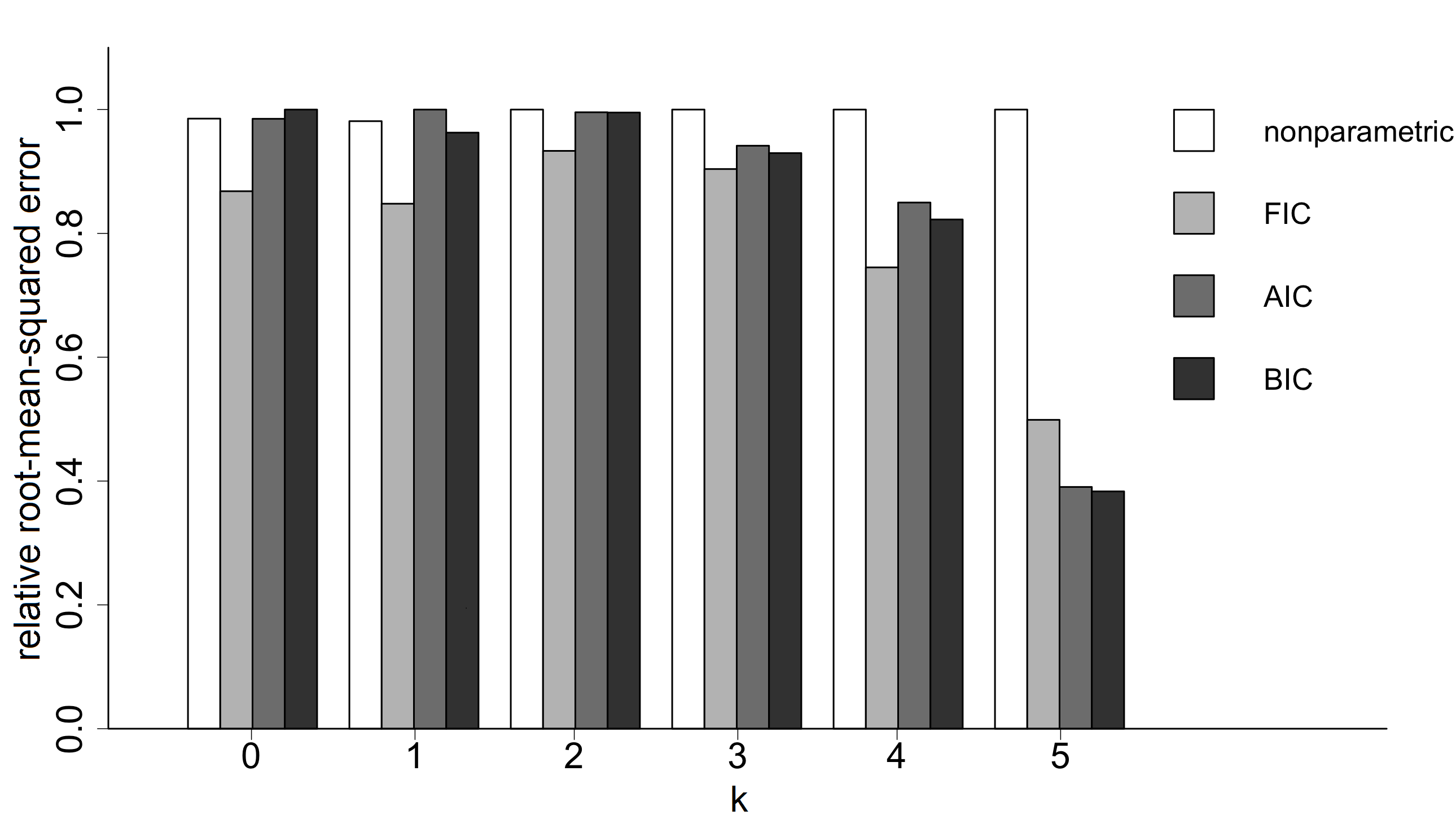}
\caption{
The relative root-mean-squared (computed in the same 
simulations) for the models selected by FIC, AIC and BIC, 
and by always using the nonparametric model.}
\label{figure:simulation2}
\end{figure}

\section{Concluding remarks}
\label{section:concludingremarks}

Here we offer a list of conclucing comments, some pointing
to further relevant research. 

\subsection{Model averaging}

The FIC scores may also be used to combine the most promising
estimators into a model averaged estimator, say 
$\hatt\mu^* = \sum_j c(M_j)\hatt\mu_j$,
with $c(M_j)$ given higher values for models $M_j$ with higher FIC scores; 
as discussed in \citet{HjortClaeskens03}.

\subsection{The conditional FIC}

For time series processes, several interesting and important 
focus parameters are naturally related to predictions, are sample size 
dependent or otherwise formulated conditional on past observations. The classical 
example is $k$-step ahead predictions. A class of such estimands could take the form 
$$\mu(\alpha,\gamma,y_{1}, \ldots, y_{m}) 
   = \Pr\{Y_{n + 1} > \alpha \textrm{ and } 
  Y_{n + 2} > \gamma\midd y_{1}, \ldots, y_{m} \} $$ 
for a suitable choice of $\alpha$.
The dependency on previous data requires a new and extended 
modelling framework, which in \citet[Sections 5 \& 6]{HermansenHjort15c} 
led to generalisations and also motivated a 
conditional focused information criterion (cFIC). 
In completing the FIC-framework for selecting among parametric 
and nonparametric time series models, such considerations
should also be taken into account.

\subsection{Linear time series processes}

Building on \cite{walker64,Hannan73,Brillinger75}, the main results of 
Section \ref{section:parametricornonparametric} can be 
extended to more general types of time series processes, 
like the generalised linear processes (cf.~\citet{Priestley81}); 
also without the assumption of Gaussian innovation terms.

\subsection{Trends and covariates}
In the presented work, our focus was on the dependency 
structure only. However, the methods and results of our paper 
may be generalised to select simultaneously among models 
with different trends and dependency structures, like 
$Y_t=m(x_{t} ,\beta) + \eps_t$, with $\eps_t$ a stationary 
Gaussian time process. These issues, leading to a larger 
repertoire of FIC formulae, will be returned to in later work. 
Since it is generally hard to estimate both the trend and dependency 
structure using a full nonparametric framework, the two main challenges
is to extend the existing work to handle the case with various 
parametric candidates for the trend $m(x_{t},\beta) $
and both parametric models and a nonparametric candidate for the 
dependency, i.e.~the spectral distribution (since we are working under 
the Gaussian assumption). Alternatively, we may assume that the 
$\eps_t$ belongs to an appropriate width family of parametric stationary 
time series processes, such as the autoregressive AR, the moving average 
MA or the mixture ARMA (cf.~\cite{Brockwell91}) and instead compare  
a nonparametric method for estimating the trend part of the model, 
perhaps extending this to functions of the type $m(t, x_{i}, \beta)$,
against a class of parametric alternatives. 

\subsection{The local large-sample framework}
\label{sec:local}

As mentioned in the introduction, \cite{HermansenHjort15c} derives
FIC for selecting among parametric time series models using a local 
asymptotics framework. 
The parametric candidate models then have spectral densities
belonging to a parametric family 
$f(\cdot;\theta,\gamma)$, with a $p$-dimensional protected 
$\theta$ and a $q$-dimensional open $\gamma$. This constitutes
a set of $2^q$ potential parametric candidate models.
The full (or wide) model is represented by the spectral
density $f(\cdot;\theta,\gamma)$. At the other 
end of the spectrum, the narrow model 
corresponds to fixating $\gamma = \gamma_{0}$, a known 
value, with the resulting $f(\cdot; \theta) 
= f(\cdot;\theta,\gamma_{0})$. 
The local misspecification framework then assumes that 
the true spectral density takes the form
$f(\cdot;\theta_0,\gamma_0+\delta/\rootn)$,
for some unknown $q$-dimensional $\delta$ 
describing the distance to the wide model.
This framework causes variances and squared biases to become
of the same order of magnitude $O(1/n)$. Those lead
to approximation formulae for the mean squared error and FIC formulae 
for nested parametric models, which are different from those obtained in this paper.

The introduction of the `asymptotically correct' nonparametric model of the present 
paper allowed us to derive FIC formulae even when sidestepping the above local misspecification assumption.
An alternative approach is to retain the local asymptotics framework
and work with spectral densities 
of the type 
$f_{r}(\omega) = f_{\theta_{0}}(\omega) + r(\omega) / \rootn$, where 
$f_{\theta_{0}}$ is a standard type of parametric 
model. 
Such structures have already been worked with in \cite{Dzhaparidze86}, making the extension potentially 
less cumbersome. 
This will not be dealt with here, however.


%



\bibliographystyle{biometrika}
\bibliography{hhj_bibliography3}

\end{document}